\documentclass[twocolumn,secnumarabic,amssymb, amsmath, nofootinbib,tightenlines,
nobibnotes, aps,
prl]{revtex4-1}
\makeatletter

\newcommand{\Rmnum}[1]{\expandafter\@slowromancap\romannumeral #1@}
\makeatother

\newtheorem{theorem}{Theorem}[section]

\newenvironment{proof}[1][Proof]{\begin{trivlist}
\item[\hskip \labelsep {\bfseries #1}]}{\end{trivlist}}

\newcommand{\qed}{\nobreak \ifvmode \relax \else
      \ifdim\lastskip<1.5em \hskip-\lastskip
      \hskip1.5em plus0em minus0.5em \fi \nobreak
      \vrule height0.75em width0.5em depth0.25em\fi}
\begin{document}

%\preprint{}

\title{Quantum key distribution with limited classical Bob}


\author{Zhi-Wei Sun$^1$}
\email[]{sunzhiwei1986@gmail.com}
\author{Rui-Gang Du$^1$}
\email{Duruigang@yahoo.com.cn}
%\author{Bang-Hai Wang$^1$}
\author{Dong-Yang Long$^1$}
\email{issldy@mail.sysu.edu.cn}

\affiliation{
$^{1}$School of Information Science and Technology, Sun Yat-sen University, 510006, P.R.China
}

\date{\today}

\begin{abstract}
Two QKD protocols with limited classical Bob who is restricted to only preparing a qubit in the classical basis and sending it or doing nothing are presented, and they are proved completely robust. As limited classical Bob can deterministically prepare the qubit, we exploit this feature to construct a quantum secure direct communication protocol with limited classical Bob, which is the direct communication of secret messages without first producing a shared secret key.
\end{abstract}

%\keywords{}

\maketitle

\section{\Rmnum{1}. INTRODUCTION}
Since the pioneering work of Bennett and Brassard (BB84) \cite{BB84}, quantum key distribution (QKD) has been developed quickly and has started approaching maturity, ready for implementation in realistic setting \cite{2009Scarani}. QKD mainly realizes the transmission of secret messages, which is an important task in cryptography; and the security of the QKD does not depend on the intractability of mathematical task, which can be solved efficiently if quantum computer could be available someday, but the general principles of quantum mechanics. The quantum mechanics can guarantee the security of the protocols without imposing any restriction of the power of the eavesdropper.
However, previous quantum key distribution protocols require all parties to have quantum capabilities which may be unpractical in various applications because not all the participants can afford expensive quantum resources and quantum operations. So the question of how "quantum" a protocol should be in order to achieve a significant advantage over all classical protocol is of great interest \cite{2007Boyer}.
In 2007, Boyer et al. presented, for the first time, a protocol in which one party (Bob) is classical \cite{2007Boyer} to answer this question in the field of quantum cryptography. We call such a protocol BKM2007 for short. This "semi-quantum" protocol has advantages over fully quantum protocol because it is easier to implement in practice, and maintains all the advantages of the original system. Now, there have been some of the results in the area of semi-quantum protocols \cite{2009Boyer,2009Zou,2010Li}.

The setting of BKM2007 is as follow \cite{2007Boyer}: Alice and Bob have labs that are perfectly secure; they use qubits for their quantum communication; they have access to an authenticated public classical communication channel; a quantum channel leads from Alice's lab to the outside world and back to her lab; Bob can access a segment of the channel, and whenever a qubit passes through that segment Bob can either let it go undisturbed or $(\overline{1})$ measure the qubit in the fixed orthogonal basis set $\{|0\rangle, |1\rangle\}$ which is also called classical, and $(\overline{2})$ prepare a (fresh) qubit in the classical basis and send it.
Bob is classical if he is limited to performing only operations $(\overline{1})$ and $(\overline{2})$ or doing nothing, and he cannot obtain any quantum superposition of the two states in the classical basis.

To answer the question of how "quantum" a protocol should be in order to achieve a significant advantage over all classical protocol.
It is natural to ask whether there exists a high-secure protocol in which classical Bob is limited to perform only operation $(\overline{2})$ or do nothing, we call such Bob limited classical Bob for convenience.
There are four reasons why we consider the limited classical Bob. First of all, it gives an answer to the question of how "quantum" a protocol should be in order to achieve a significant over all classical protocol. Second, not all the participants can afford expensive quantum operations, for example, measuring the quantum qubits using qubit detector. So the protocol with limited classical Bob may be more practical in various applications. Third, in principle, quantum key distribution offers unconditional security where assumptions are made for the devices involved. However, in practical implementations the components deviate from the models in the security proofs. Eve could exploit imperfections in Alice's or Bob's equipment (such as source or detectors) remotely to acquire their secret key. Several such attacks have been proposed \cite{2007Fung,2008Zhao}. Recently, it has been experimentally shown that the detectors of two commercially available QKD systems can be fully remote-controlled using specially tailored bright illumination, which make it possible to tracelessly obtain the full secret key \cite{2010Lydersen}. The protocol with limited classical Bob who is limited not using the detectors may be secure against this implementation-dependent attack. So it is more secure to realize and be favored in some specific applications. Last,
the semi-quantum key distribution protocols \cite{2007Boyer,2009Boyer,2009Zou} are usually nondeterministic, i.e., Alice can encode a classical bit into a quantum state, which is then sent to Bob, but she cannot determine the bit value that Bob will finally decode. But such nondeterministic communication can be used to establish a shared secret key between Alice and Bob, just as done in quantum key distribution protocol. However, the protocol with limited classical Bob can also be used to deterministically send message through quantum channel, which is also known as quantum secure direct communication \cite{2002首篇Lett.89.187902}.

In this paper, we present two QKD protocols with limited classical Bob, and they also allow secure direct communication. Furthermore, we prove that all the protocols are completely robust. Robustness of a protocol means that any adversarial attempt to learn some information on the key necessarily induces some disturbance \cite{2007Boyer}.
To prove a protocol being robust is an important step in studying security, and Boyer et al \cite{2007Boyer} particularly divided robustness into three classes: completely robust, partly robust and completely nonrobust. A protocol is said to be completely robust if nonzero information acquired by Eve on the INFO string (before Alice and Bob perform the ECC step) implies nonzero probability that the legitimate participants find errors on the bits tested by the protocol. A protocol is said to be partly robust if Eve can acquire some limited information on the INFO string without inducing any error on the bits tested by the protocol. A protocol is said to be completely nonrobust if Eve can acquire the INFO string without inducing any error on the bits tested by the protocol. Partly robust protocols could still be secure, yet completely nonrobust protocols are automatically proven insecure \cite{2007Boyer}. As one example in Ref. \cite{2007Boyer}, BB84 is completely robust when qubits are used by Alice and Bob but it is only partly robust if photon pulses are used and sometimes two-photon  pulses are sent.

This paper is organized as follows. In sec. \Rmnum{2} and sec. \Rmnum{3}, we present two QKD protocols with limited classical Bob. In section. \Rmnum{4}, a quantum secure direct communication protocol is proposed. And finally, we give the conclusion.
\section{\Rmnum{2}. SQKD PROTOCOL IN WHICH ALICE SENDS FOUR QUANTUM STATES TO LIMITED CLASSICAL BOB}
We follow the BKM2007 protocol's idea\cite{2007Boyer} to construct an SQKD protocol described in the following, but the classical Bob is restricted to performing only operation $(\overline{2})$ or doing nothing.
\subsection{SQKD Protocol 1: Alice sends four quantum states to limited classical Bob}
Let the integer $n$ be the desired length of the INFO string, and let $\delta >0$ be some fixed parameter.

$(1)$ Alice randomly creates $N=8n(1+\delta)$ qubits, each of which is either in the computational ("Z") basis $\{|0\rangle, |1\rangle\}$ or in the diagonal ("X") basis $\{|+\rangle, |-\rangle\}$, where $|+\rangle =\frac{1}{\sqrt{2}}(|0\rangle + |1\rangle)$, $|-\rangle = \frac{1}{\sqrt{2}}(|0\rangle - |1\rangle)$. And Alice sends the $N$ qubits to Bob. After Alice sends the first qubit, she sends a qubit only after receiving the previous one \cite{2007Boyer}.

$(2)$ When each qubit arriving, Bob chooses randomly either to reflect it (CTRL) or randomly prepare a fresh qubit in the classical basis and resend it (SIFT). Bob resends a qubit immediately after receiving it.

$(3)$ Alice measures each qubit in the basis she sent it.

$(4)$ Alice announces which were her $Z$ bits and Bob publishes which ones he chose to SIFT. It expected that approximate $\frac{N}{4}$ bits, Alice used the $Z$ basis for transmitting and Bob chose to SIFT; these are the SIFT bits, which form the sifted key. They abort the protocol if the number of SIFT bits is less than $2n$; this happens with exponentially small probability.

$(5)$ Alice checks the error rate on the CTRL bits (CTRL bits denote the bits produced by the process that Bob chooses to CTRL). If either the $Z-CTRL$ (Alice used the $Z$ basis and Bob chose CTRL) error rate or the $X-CTRL$ (Alice used the $X$ basis and Bob chose CTRL) error rate is higher than some predefined threshold $P_{CTRL}$, she and Bob abort the protocol.

$(6)$ Bob chooses at random $n$ sifted key to be TEST bits. He publishes which are the chosen bits. Alice publishes the value of these TEST bits. Bob check the error rate on the TEST bits and if it is higher than some predefined threshold $P_{TEST}$, they abort the protocol.

$(7)$ Alice and Bob select the first $n$ remaining sifted key to be used as INFO bits (INFO string).

$(8)$ Bob announces error correction code (ECC) and privacy amplification (PA) data; he and Alice use them to extract the $m$-bit final key from the $n$-bit INFO string.

\begin{theorem}
The SQKD Protocol $1$ is completely robust.
\end{theorem}
\begin{proof}
The difference between this protocol and BKM2007 protocol is that for the process of SIFT, no measurement is made and a fresh qubit is randomly prepared in the classical basis, whereas a fresh qubit is prepared in the classical basis according to the measurement result in the BKM2007 protocol. However, the fresh qubits are prepared in a random manner in both cases. The security of BKM2007 is assured by means of Alice and Bob checking the error rate of CTRL bits and TEST bits. In fact the proof of complete robustness of SQKD Protocol $1$ is identical to that in the BKM2007 protocol. Because the BKM2007 has been proven completely robust \cite{2007Boyer}, the complete robustness of SQKD Protocol $1$ is assured.
\end{proof}
\section{\Rmnum{3}. SQKD PROTOCOL IN WHICH ALICE SENDS ONLY ONE QUANTUM STATE TO LIMITED CLASSICAL BOB}
In 2009, a nice extension of BKM2007 was suggested by Zou et al. \cite{2009Zou}, which suggests that it is sufficient for the originator of the states (the person holding the quantum technology) to generate just one state. We call such a protocol ZQLWL2009 for short. Fortunately we can also follow the ZQLWL2009's idea of SQKD to construct an SQKD protocol with limited classical Bob in which Alice sends only one quantum state described as follows.
\subsection{SQKD Protocol 2: Alice sends one quantum state to limited classical Bob}
$(1)$ Alice creates and sends $N$ qubits $|+\rangle^{N}$, where $N=8n(1+\delta)$, $n$ be the desired length of the INFO string, and $\delta >0$ be a fixed parameter. After Alice sends the first qubit, she sends a qubit only after receiving the previous one.

$(2)$ When each qubit arriving, Bob chooses randomly either to reflect it (CTRL) or randomly prepare a fresh qubit in the classical basis and resend it (SIFT). Bob resends a qubit immediately after receiving it.

$(3)$ Alice randomly measures each qubit either in the $Z$ basis or in the $X$ basis.

It is expected that for approximately $\frac{N}{4}$ bits, Bob chooses to SIFT and Alice measures in the $Z$ basis; these are the SIFT bits, which form the sifted key. For approximately $\frac{N}{4}$ bits, Bob chooses to CTRL and Alice measure in the $X$ basis; we refer to these bits as $CTRL-X$.

$(4)$ Alice announces which basis she chooses to measure the qubits and Bob publishes which ones he chose to SIFT. They check the number of sifted key. They abort the protocol if the number of sifted key is less than $2n$.

$(5)$ Alice checks the error rate on the $CTRL-X$. If the error rate of $CTRL-X$ is higher than some predefined threshold $P_{t}$, she and Bob abort the protocol.

$(6)$ Bob chooses at random $n$ sifted key to be TEST bits. He publishes which are the chosen bits. Alice publishes the value of these TEST bits. Bob checks the error rate on the TEST bits and if it is higher than some predefined threshold $P_{TEST}$, the protocol abort.

$(7)$ Alice and Bob select the first $n$ remaining sifted key to be used as INFO bits.

$(8)$ Bob announces ECC and PA data; he and Alice use them to extract the $m$-bit final key from the $n$-bit INFO string.
\begin{theorem}
The SQKD Protocol 2 is completely robust.
\end{theorem}
\begin{proof}
It is similar to the proof of complete robustness of the protocol in Ref. \cite{2011Comment}.
\end{proof}

\section{\Rmnum{4}. QUANTUM SECURE DIRECT COMMUNICATION PROTOCOL WITH LIMITED CLASSICAL BOB }
One might see a new feature in our protocol (with limited classical Bob) over some protocols (for example, \cite{2007Boyer}) in the fact that the bits are not random but can be chosen deterministically by classical Bob. Having this properties, the protocol is not restricted to key distribution only; it can be used for quantum secure direct communication \cite{2002首篇Lett.89.187902}, which is the direct communication of secret messages without first producing a shared secret key.

At first glance, it seems to be a simple work to constitute a quantum secure direct communication (QSDC) protocol following the idea of SQKD protocol with limited classical Bob. However, it is completely insecure because Eve can steal the secret message with man-in-the-middle attack strategy. Eve stores all the qubits coming back from classical Bob and resends randomly forged qubits to Alice. After Bob announces which qubit he chose to SIFT, Eve measures qubits of SIFT in the computational basis and obtains the secret message. Fortunately, the method used in Ref. \cite{1995确定分发} can be used in our protocol to against the man-in-the-middle attack, and quantum register is used. Let us give an explicit description of the protocol as follows.

$(1)$ Alice and Bob agree that $|0\rangle$ and $|1\rangle$ represent $0$ and $1$ respectively.

$(2)$ In order to transmit a message of some length $n$, Bob constructs a longer string: some extra bits are used for estimating the error rate (hence, the maximal information eavesdropped by Eve) and some for redundancy, which is used (via block coding) to encode the $n$-bit message. We denote the longer string as $m$ for convenience, and the length of $m$ is $|m|$. Bob notices Alice that he wants to send a secret message to her.

$(3)$  Alice generates random $N= 2 |m|(1+\delta)$ qubits either in the computational basis $\{|0\rangle, |1\rangle\}$ or in the diagonal basis $\{|+\rangle, |-\rangle\}$, and she sends them to classical Bob. After Alice sends the first qubit, she sends a qubit only after receiving the previous one.

$(4)$  When each qubit arriving, Bob chooses randomly either to reflect it (CTRL) or prepare a fresh qubit in the classical basis according to the message $m$ and resend it (SIFT). Bob resends a qubit immediately after receiving it.

$(5)$ Alice uses an $N$-qubit register to save all qubits coming back from Bob.

$(6)$ After verifying that Alice has received all $N$ qubits, Bob announces which ones he chose to SIFT.

$(7)$ Alice measures each qubit of CTRL in the basis she sent it and each qubit of SIFT in the computational basis. She checks the error rate on the CTRL bits. If either the $Z-CTRL$ (Alice used the $Z$ basis and Bob chose CTRL) error rate or the $X-CTRL$ (Alice used the $X$ basis and Bob chose CTRL) error rate is higher than some predefined threshold $P_{CTRL}$, she and Bob abort the protocol.

$(8)$ Bob tells Alice which bits were used for error estimation on the SIFT bits. If Alice, estimating the error rate, detects Eve, she prevents public announcement of the block-coding function by informing Bob. Thus the secret message can be transmitted with an exponentially small probability of errors and exponentially small information leakage.
\begin{theorem}
The QSDC protocol with limited classical Bob is completely robust.
\end{theorem}
\begin{proof}
It is straightforward by the proof of SQDK Protocol $1$ that Alice can get the longer string $m$ securely. If Eve uses the man-in-the-middle attack described above, Alice estimating the error rate in the step $(8)$, will detect Eve, and she informs Bob. Then Bob will stop public announcement of the block-coding function. Eve won't get the $n$-bit message without knowing the block-coding function.
\end{proof}
\section{\Rmnum{5}. CONCLUSION}
In this paper, we presented two QKD protocols with limited classical Bob who performs only limited classical operations (preparing a (fresh) qubit in the classical basis and send it or doing nothing) and proved its robustness. The presented two protocols have the feature that the bits can be chosen deterministically by classical Bob. Using this feature, we also constructed a quantum secure direct communication. These give a new answer to how much "quantumness" is required in order to perform classically impossible tasks in general \cite{2007Boyer}.
\section{\Rmnum{6}. ACKNOWLEDGMENTS}
This work is in part supported by the Key Project of NSFC-Guangdong
Funds (No.U0935002).

\end{document}